\documentclass[journal,twoside,web]{ieeecolor}
\usepackage{lcsys}

\usepackage{multicol}
\usepackage[bookmarks=true]{hyperref}
\usepackage{lipsum}
\usepackage[caption=false,font=footnotesize]{subfig}

\usepackage{amssymb} 

\usepackage{bm}

\usepackage{enumerate}
\usepackage{commath}
\usepackage{graphicx}
\usepackage{amsmath}
\usepackage{mathtools}

\usepackage{breqn}
\usepackage{booktabs}
\usepackage{empheq}
\usepackage{amsfonts}

\usepackage{amsthm}
\usepackage{mathrsfs}
\usepackage{accents}
\usepackage{thmtools}
\usepackage{thm-restate}
\usepackage{float}
\usepackage[normalem]{ulem}

\usepackage{algorithm}
\usepackage{algpseudocode}
\usepackage{dblfloatfix}

\usepackage{etoolbox}
\usepackage{float}  

\usepackage{comment}

\usepackage{calc}
\newlength{\mywidth}

\theoremstyle{plain}

\newtheorem{lemma}{Lemma}

\newtheorem{proposition}{Proposition}

\newtheorem*{theorem*}{Theorem}
\newtheorem{assumption*}{Assumption}

\declaretheorem[name=Theorem]{thm}

\theoremstyle{definition}
\newtheorem{example}{Example}
\newtheorem{remark}{Remark}
\newtheorem{definition}{Definition}
\newtheorem{assumption}{Assumption}

\newtheorem*{problem*}{Problem}

\newcommand{\R}{\mathbb{R}}

\newcommand{\myvar}[1]{#1}

\newcommand{\myset}[1]{\mathcal{#1}}

\DeclareMathOperator*{\argmin}{argmin}

\newenvironment{examplecont}[1]
  {\par\medskip\noindent\textbf{Example~#1 Continued.}\normalfont}
  {\par\medskip}

\title{A Duality-Based Optimization Formulation of Safe Control Design with State Uncertainties}

\author{
Xiao Tan$^*$, Rahal Nanayakkara$^*$, Paulo Tabuada,  and Aaron D. Ames
\thanks{This work is supported by TII under project \#A6847.}
\thanks{Xiao Tan and Aaron D. Ames are with the the Department of Mechanical and Civil Engineering, California Institute of Technology, Pasadena, CA 91125, USA (Email: {\tt\small xiaotan,  ames@caltech.edu}).}    
\thanks{Rahal Nanayakkara and Paulo Tabuada are with the Department of Electrical and Computer Engineering at University of California, Los Angeles, CA 90095, USA (Email: {\tt\small  rahaln@g.ucla.edu, tabuada@ucla.edu}).} 
\thanks{$^*$ Equal contributions.}
}

\begin{document}

\maketitle
\thispagestyle{empty}
\pagestyle{empty}

\begin{abstract}
State estimation uncertainty is prevalent in real-world applications, hindering the application of safety-critical control. 
Existing methods address this by strengthening a Control Barrier Function (CBF) condition either to handle actuation errors induced by state uncertainty, or to enforce stricter, more conservative sufficient conditions.
In this work, we take a more direct approach and formulate a robust safety filter by analyzing the image of the set of all possible states under the CBF dynamics.
We first prove that convexifying this image set does not change the set of possible inputs.
Then, by leveraging duality, we propose an equivalent and tractable reformulation for cases where this convex hull can be expressed as a polytope or ellipsoid.
Simulation results show the approach in this paper to be less conservative than existing alternatives.
\end{abstract}


\section{Introduction}
\IEEEPARstart{S}{afety} is a crucial requirement in modern control applications. 
Formally, safety is a property that requires all trajectories of the system to remain within a prescribed safe set. 
Recently, Control Barrier Functions (CBF)~\cite{Ames2017} have emerged as a convenient tool to enforce safety guarantees for control systems
by constructing an inequality whose satisfaction implies the assurance of system safety.
Furthermore, for systems which are affine in the control input, the CBF inequality reduces to a linear constraint on the input, allowing fast and easy computation of safe control inputs in real time.

Despite its effectiveness, the standard CBF framework assumes perfect knowledge of system dynamics and state information. However, in practical applications, state uncertainty may arise from sensor degradation, environmental interference, or ill-calibrated state estimation algorithms. Fortunately, advances in observer design provide not only state estimates, but also deterministic bounds on the estimation error \cite{alamo2005zonotope, jaulin2002estimation, silvestre2024nonlinear}.
This paper addresses the problem of finding a safe control input given 
a set of states guaranteed to contain the true state, which can be obtained using these observer bounds.

Enforcing safety under these conditions requires satisfying the CBF inequality over the entire set of possible states. 
Measurement Robust CBFs (MR-CBF)\cite{cosner2021measurement} achieve this by bounding each state dependent coefficient appearing in the CBF inequality via their respective Lipschitz constants.
However, this method suffers from being overly conservative and the feasibility of such constraints is largely left unanswered.
This MR-CBF framework is generalized to include uncertainty in the dynamics in \cite{lindemann2024learning}, and expanded via interval analysis with higher-order Taylor models in \cite{zhang2022control}.
The result in \cite{agrawal2022safe}, uses a similar technique but also assumes that the control effectiveness gain in the CBF inequality does not change sign, heavily restricting the applicability of this method.

An alternative perspective considers the effect of state uncertainty on the closed loop control law, providing guarantees against the resulting actuation errors \cite{jankovic2018robust, nanayakkara2025safety, kolathaya2018input, das2025safe}. 
We discuss this further in Section \ref{sct:dynamic_uncertainty}. 
For stochastic state uncertainties, works like \cite{cosner2023robust, ramadan2024control, do2024probabilistically} provide probabilistic safety guarantees, although we do not make direct comparisons as our approach focuses on the deterministic setting.

In this work, we take a direct approach to the safe control problem under state uncertainty by
designing an optimization-based safety filter that enforces a CBF constraint across the entire set of possible states. 
This is naturally formulated as a robust optimization problem.
We first show that replacing the image set of the map from the state uncertainty to the CBF coefficients with its convex hull does not change the set of feasible solutions for the robust optimization problem.
We next focus on the case where this convex hull can be expressed as a polytope or ellipsoid, and propose a duality-based reformulation that is both equivalent to the original problem and efficient to solve online.
In contrast to prior works such as \cite{cosner2021measurement}, which rely on Lipschitz constants to provide a simple, albeit coarse, over-approximation of the uncertainty set, 
our approach accurately computes this required set resulting in a significantly less conservative formulation.

\section{Preliminaries}

\textit{Notation} : We denote by $\| x\|$ the Euclidean norm of $x \in \R^n $. 
For an essentially bounded signal $d(t)$, we denote its $\mathcal{L}_\infty$ norm by $\|d\|_{\infty} = \operatorname{ess\,sup}_{t \ge 0} \|d(t)\|$.
For $g:\R^n \to \R^{n \times m}$ and $h: \R^n \to \R$, $L_gh(x) = \nabla h(x) g(x)$.
Given $A, B \subseteq \R^n$, $A \oplus B$ is the Minkowski sum of $A$ and $B$. A continuous function $\alpha : \R_{\geq 0} \to \R_{\geq 0}$ (respectively, $\alpha : \R \to \R$) is said to be of class $\mathcal{K}$ (respectively, extended class $\mathcal K$) if $\alpha(0)=0$ and $\alpha$ is strictly monotonically increasing.

Consider a control-affine system:
\begin{equation} \label{eq:system}
    \dot{\myvar{x}} = f(\myvar{x}) + g(\myvar{x})\myvar{u} ,
\end{equation}
where $\myvar{x}\in \mathbb{R}^n$, $\myvar{u} \in U \subset \mathbb{R}^m$ are the system state and input, respectively. 
We assume that $U$ is a polytope
\footnote{This duality-based technique applies to any $U \subseteq \R^m$, though other sets may make the resulting problem more challenging to solve online.},
and the vector fields $f, g$ are locally Lipschitz continuous. 

\subsection{Control Barrier Functions}
 Suppose that the system state is required to be confined to a region $\myset{C} = \{\myvar{x}\in \mathbb{R}^n: h(\myvar{x})\geq 0\}$, 
 where $h:\mathbb{R}^n \to \mathbb{R}$ is a continuously differentiable function.
 We say that the set $\myset{C}$ is safe if it is forward invariant. One way to enforce this state constraint is through Control Barrier Functions (CBFs).

\begin{definition}[Control Barrier Functions~\cite{Ames2017}]
    Given the nonlinear system in \eqref{eq:system}, a continuously differentiable function $h$ is a CBF if there exists an open set $\myset{D}\supset \myset{C}$ and an extended class $\mathcal{K}$ function $\alpha(\cdot)$ such that
    \footnote{The strict inequality ensures continuity and robustness of the controller. See \cite{jankovic2018robust} for more details.}:
    \begin{equation} \label{eq:cbf_def}
        \sup_{u\in U } 
        \underbrace{L_{f} h (\myvar{x}) + L_{g} h (\myvar{x}) u}_{\dot h(x,u)}  > -\alpha(h(\myvar{x})), \quad \forall \myvar{x}\in \myset{D}.
    \end{equation}
\end{definition}
Any locally Lipschitz control law $u = k(x)$ satisfying:
 \begin{equation} \label{eq:cbf_condition}
     L_{f} h (\myvar{x})+ L_{g} h (\myvar{x}) u  \geq -\alpha(h(\myvar{x})),
 \end{equation}
 renders the set $\myset{C}$ forward invariant \cite{Ames2017}. 
 Given some desired (potentially unsafe) controller $k_d : \R^n \to \R^m$, the minimally invasive safe control law can be computed as in \cite{Ames2017}:
\begin{equation} \label{eq:CBF QP}
    \begin{aligned}
    k(x)  = & ~\underset{u \in U}{\textup{argmin}}
    \left\Vert u - k_d({x})  \right\Vert^{2} \\
    \textup{s.t.} \  &
     L_fh(x) + L_gh(x)u+\alpha(h(x)) \geq 0.
    \end{aligned}
\end{equation}

\subsection{A Dynamic Uncertainty Perspective} \label{sct:dynamic_uncertainty}

In practice, feedback control relies on a state estimate $\hat{x}$ rather than the true state $x$. Given a feedback control law $k : \mathbb R ^n \to \mathbb R ^m$, applying the control input $u = k(\hat{x})$ instead of the ideal $u=k(x)$ introduces a dynamic disturbance \mbox{$d = k(\hat{x}) - k(x)$}, resulting in the closed-loop dynamics:
\begin{equation} \label{eq:disturbed_system}
\dot{\myvar{x}} = f(\myvar{x}) + g(\myvar{x})(k(x) + d).
\end{equation}
When $\hat{x} - x$ is bounded and the controller $k$ is locally Lipschitz, {for any bounded $x$} one concludes that $d$ is bounded. 
Here, $d$ acts as a matched dynamic uncertainty to the nominal closed-loop system \mbox{$\dot{\myvar{x}}  = f(\myvar{x}) + g(\myvar{x})k(x) $}. A well-established safety notion in this setting is called Input-to-State Safety (ISSf) \cite{kolathaya2018input}.

A controller satisfying an ISSf-CBF constraint ensures the forward invariance of an inflated safe set:
\begin{equation} \label{eq:inflated safe set}
    \myset{C}_{d} = \{x\in \mathbb{R}^n: h(\myvar{x}) + \gamma(\| d\|_\infty) \geq 0\},
\end{equation}
under the perturbed dynamics \eqref{eq:disturbed_system}, where $\gamma$ is a class $\mathcal K$ function. Intuitively, an ISSf-CBF guarantees that the system remains within a neighborhood of the original safe set that scales with the magnitude of the disturbance.

 To achieve ISSf style guarantees, the recent work \cite{nanayakkara2025safety} proposes the R-CBF constraint:
\begin{equation} \tag{R-CBF} \label{eq:R-CBF}
L_f h(x) + L_g h(x)u + \alpha(h(x)) \geq 
\gamma_{1}\|L_{g} h(x)\| + \gamma_{2}\|L_{g} h(x)\|^{2},
\end{equation}
with tuning parameters $\gamma_1, \gamma_2 \in \R_{>0}$. The additional terms effectively compensate for the worst-case impact of bounded disturbances.
As established in \cite{nanayakkara2025safety}, $\gamma_1$ and $\gamma_2$ dictate the disturbance threshold determining whether the original safe set $\myset{C}$ or an inflated set \eqref{eq:inflated safe set} is rendered forward invariant. Although recent extensions adapt these parameters online \cite{das2025safe}, it remains unclear how to systematically design $\gamma_1$ and $\gamma_2$ to optimize the closed-loop robustness/safety guarantees when the full state uncertainty set is known.

\begin{example}
    Consider the scalar system from \cite{nanayakkara2025safety}:
    \begin{equation} \label{eq:ex_scalar_system}
        \dot{x} = x(x-1.05)(x+1.05) + (1 - x^2)u,
    \end{equation}
    where $x,u\in \mathbb{R}$, and the CBF is $h(x) = 1 - x^2$ with $\alpha(h(x)) = h(x)$. 
    We study the scenario when the state estimate is $\hat x=1$, but the true state is unknown and lies in $[0.95, 1.05]$.
    
    Now, let us compare several existing results in \cite{agrawal2022safe,cosner2021measurement,nanayakkara2025safety}.
    For \cite{agrawal2022safe}, since $ L_gh(x)$ changes sign when $ x\in [0.95, 1.05]$, the Assumption~2 therein fails and thus the method becomes inapplicable.  
    The work in \cite{cosner2021measurement} imposes the MR-CBF condition: 
    \begin{multline} \tag{MR-CBF}
        L_f h(\hat{x}) + L_gh(\hat{x})u + \alpha(h(\hat{x})) \geq \\
         \epsilon(\hat x)(\mathcal{L}_{L_f h} + \mathcal{L}_{\alpha \circ h} + \mathcal{L}_{L_gh}\| u\| ),
         \label{eq:MR-CBF}
    \end{multline}
    where $\epsilon(\hat x) \geq \|x - \hat x \|$ is the maximum allowable level of state uncertainty, and $\mathcal{L}_{L_f h}, \mathcal{L}_{\alpha \circ h}, \mathcal{L}_{L_gh}$ are the Lipschitz constants of the functions in their respective subscripts.
    This MR-CBF condition turns out to be infeasible in this case, since $L_gh(\hat{x})=0$ and substituting the Lipschitz constants into the inequality reduces to the condition \mbox{$\|u\| < 0$.}
    In contrast, the R-CBF method in \cite{nanayakkara2025safety} is always feasible, but computing the threshold where the guarantees degrade from preserving the original safe set to an inflated one is difficult. 
    Moreover, even when the R-CBF only provides guarantees for the inflated safe set, there may still exist an input $u$ that renders the original set invariant.
    In what follows, we propose a direct approach to analyze the exact conditions under which the CBF constraint is satisfied over the entire set of possible states.
\end{example}

\subsection{Problem Formulation}

We denote by $\myset{B}$ the compact set containing all possible estimation errors, obtained using bounded error observers such as \cite{alamo2005zonotope, jaulin2002estimation, silvestre2024nonlinear}.
Naturally, this set must also contain the true estimation error:
\begin{equation} \label{eq:state uncertainty}
   e = x - \hat{x} \in \myset{B}.
\end{equation}
The set $\myset{B}$ may in fact be a function of both the estimate $\hat x$ and time $t$, (i.e., a time and state dependent error bound) though we do not denote this explicitly as it does not affect our arguments.
Common state uncertainty sets include hyper-rectangles \cite{alamo2005zonotope, jaulin2002estimation} and norm balls \cite{silvestre2024nonlinear}.

In this work, we aim to design a safe feedback controller with knowledge of $\hat{x}$ and $\myset{B} $. Since the exact true state is unknown, one approach to ensure the CBF inequality is satisfied at the true state $x$ is to enforce the constraint in (\ref{eq:CBF QP}) for 
all possible states consistent with the measurement uncertainty, giving rise to the following robust optimization problem: 
\begin{equation} \label{eq:robust CBF QP}
    \begin{aligned}
    & k_{\rm ro}(\hat{x},\myset{B})  = \underset{u \in U}{\textup{argmin}}
    \left\Vert u - k_d(\hat{x})  \right\Vert^{2} \\
    \textup{s.t.} \  &
     L_{f} h(\xi)+L_{g} h(\xi)u +\alpha(h(\xi)) \geq 0, \ \ \forall \xi \in \{\hat{x}\} \oplus \myset{B}.
    \end{aligned}
\end{equation}
 A straightforward result on safety is as follows:

\begin{thm}
    Consider the nonlinear system in \eqref{eq:system} with state estimate $\hat{x}$ and estimation error set $\myset{B}$ in \eqref{eq:state uncertainty}. Suppose that $k_{\rm ro}(\hat{x},\myset{B})$ is well defined and locally Lipschitz for any $x\in \myset{D}\supset \myset{C}$.  Then, for the closed-loop system defined by \eqref{eq:system} with $ u =  k_{\rm ro}(\hat{x},\myset{B}) $, the set $\myset{C}$ is forward invariant.
\end{thm}

 \begin{proof}
    For any $x\in \myset{D}$, let its corresponding state estimate be $\hat{x}$ with the estimation error set $\myset{B}$. 
    Since $ k_{\rm ro}(\hat{x},\myset{B}) $ is a solution to \eqref{eq:robust CBF QP}, we know the CBF inequality \eqref{eq:cbf_condition}
    holds for any $x\in \myset{D}$. Thus by Proposition 1 in \cite{Ames2017}, $\myset{C}$ is forward invariant.
\end{proof}

\begin{remark}
    The local Lipschitz continuity of the proposed controller \eqref{eq:robust CBF QP} can be verified under standard regularity assumptions on $k_d$, provided the feasible set of \eqref{eq:robust CBF QP} has a non-empty interior for every $x \in \mathcal D$ (uniform Slater's condition) \cite{facchinei2003finite}.
\end{remark}

One notable issue for the control synthesis in \eqref{eq:robust CBF QP} is that for a continuous estimation error set $\myset{B}$, the robust optimization problem consists of an infinite number of linear constraints on $u$, making it intractable to solve. In the following sections, we will answer the following question: How can we efficiently compute $k_{\rm ro}(\hat{x},\myset{B})$ with little to no conservativeness?

\section{A duality-based optimization reformulation}

We first re-write our safe controller \eqref{eq:robust CBF QP} as:
\begin{equation} \label{eq:robust optimization}
       \begin{aligned}
        k_{\rm ro} (\hat{x}, \myset{B}) & = \argmin_{u\in U} f_0(u,\hat{x})  \\
       \textup{ s.t. } &  \quad a(\xi)^\top u + b(\xi) \geq 0 \quad  \textup{  for all } \xi\in \{\hat{x}\} \oplus \myset{B},
    \end{aligned} 
\end{equation}
where $f_0(u,\hat{x}) =  \left\Vert u - k_d(\hat{x})  \right\Vert^{2}$, $a(\xi) = L_{g} h(\xi)^\top$ and \mbox{$b(\xi) = L_{f} h(\xi) +\alpha(h(\xi))$}. We define the robust feasible set:
\begin{equation}
    \myset{U}_{ro} = \{u\in U: a(\xi)^\top u + b(\xi) \geq 0, \,\forall \, \xi\in \{\hat{x}\} \oplus \myset{B} \},
\end{equation}
and note that this set, formed by the intersection of (possibly infinite) affine constraints, is convex.

Standard robust optimization treats constraint coefficients $a(\cdot)$ and $b(\cdot)$ as independent uncertain parameters within sets $\myset{P}_a \subset \mathbb{R}^m$ and $\myset{P}_b \subset \mathbb{R}$ \cite{bertsimas2011theory}.
Since $b(\cdot)$ is scalar valued, the worst-case scenario over the interval $\myset{P}_b$ is trivially its minimum. Thus, the structural complexity of robust reformulations arises primarily from the uncertainty set $\myset{P}_a$. 
While efficient reformulations are well-established for polytopic or ellipsoidal $\myset{P}_a$ \cite[Section 2.2]{bertsimas2011theory}, applying them directly to \eqref{eq:robust optimization} can be conservative. 
This is because these standard methods treat the uncertainty sets independently, neglecting the physical coupling of $a(\xi)$ and $b(\xi)$ through the shared uncertain state $\xi$. We demonstrate this issue in the following example.

\begin{examplecont}{1}
   For the scalar system in \eqref{eq:ex_scalar_system}, we compute that $a(x) = 2x(x^2 - 1),\:b(x) =-2x^4 + 1.205x^2 + 1$. As before, we consider the scenario when $x \in [0.95, 1.05]$.
   The image of the function $x\mapsto (a(x),b(x))$ is a curve in $\R^m\times \R$,
   where $-0.1825 \leq a(x) \leq 0.2152$ and \mbox{$-0.1025 \leq b(x) \leq 0.4585$}.
   If we consider the uncertainty in $a$ and $b$ separately, as in \cite{bertsimas2011theory}, no feasible $u$ exists, i.e., there is no $u  \in \mathbb R$ such that:
   $$\zeta_au+\zeta_b \geq 0, \, \forall \zeta_a \in[-0.183, 0.216], \, \zeta_b \in [-0.103, 0.459].$$
   However, it can be verified that the input $u=2$ satisfies the CBF inequality for all possible states $x \in [0.95, 1.05]$.
\end{examplecont}

In the following, we derive a systematic method to construct such an input by exploring reformulations of the optimization problem \eqref{eq:robust optimization} over the uncertainty set:
\begin{equation} \label{eq:uncertainty_set}
   \myset{P}(\hat x, \myset{B}) = \{(a(\xi), b(\xi))\in \mathbb{R}^{m+1}:  \xi\in \{\hat{x}\} \oplus \myset{B}\}. 
\end{equation}
For brevity, we omit explicitly denoting the dependence of $\myset{P}$ on $(\hat x, \myset{B})$, and write \eqref{eq:robust optimization} more succinctly as:
\begin{equation} \label{eq:reformed robust optimization}
    \begin{aligned}
        \argmin_{u\in U} f_0(u,\hat{x}) 
       \textup{ s.t. } \zeta_a^\top u + \zeta_b \geq 0 \quad  
       \forall \,
       (\zeta_a, \zeta_b) \in \myset{P}.
    \end{aligned} 
\end{equation}

\subsection{Exactness of the Convex Hull Relaxation}

\begin{thm}
    The robust optimization problem in \eqref{eq:reformed robust optimization} is equivalent to the following problem:
    \begin{equation} \label{eq:reformed robust optimization2}
    \begin{aligned}
        & \argmin_{u\in U} f_0(u,\hat{x})  \\
       \textup{ s.t. } &  \quad \zeta_a^\top u + \zeta_b \geq 0 \quad  
       \forall \,
       (\zeta_a, \zeta_b) \in \tilde{\myset{P}},
    \end{aligned} 
\end{equation}
where  $\tilde{\myset{P}}$ denotes the convex hull of the set $\myset{P}$.
\end{thm}

\begin{proof}
    The two optimization problems share the same cost function. To show equivalence, it suffices to prove that any feasible $\tilde{u}$ to \eqref{eq:reformed robust optimization} is also feasible to \eqref{eq:reformed robust optimization2} and vice versa.

    Now suppose that $\tilde{u}\in U$ exists such that  \mbox{$\zeta_a^\top \tilde{u} + \zeta_b \geq 0$} for all $(\zeta_a, \zeta_b) \in \myset{P}$. For any $(\zeta_a^\prime, \zeta_b^\prime) \in \tilde{\myset{P}}$, by definition, there exists a $\lambda\in\mathbb{R}_{\ge0}^N$ with $ \sum_j\lambda_j=1$ such that 
    \mbox{$(\zeta_a^\prime, \zeta_b^\prime) = \sum_j \lambda_j (\zeta_{a,j}, \zeta_{b,j})$} 
    with $(\zeta_{a,j}, \zeta_{b,j}) \in \myset{P}$. Thus:
    \begin{equation}
    \begin{aligned}
        \zeta_{a,j}^\top \tilde{u} + \zeta_{b,j} \geq 0, \forall j 
       \implies   & \sum_j \lambda_j (\zeta_{a,j}^\top \tilde{u} + \zeta_{b,j}) \geq 0  \\
        \implies & \zeta_{a}^{\prime,\top} \tilde{u} + \zeta_{b}^{\prime} \geq 0.
    \end{aligned}
    \end{equation}
    That is, any feasible $\tilde{u}$ to \eqref{eq:reformed robust optimization} is also feasible to \eqref{eq:reformed robust optimization2}. The reverse direction holds trivially since $\myset{P}\subseteq \tilde{\myset{P}}$.
\end{proof}

The above result indicates that over-approximating the uncertainty set $\myset{P}$ with its convex hull $\tilde{\mathcal{P}}$ preserves the optimal safe input without introducing conservatism.
This motivates us to pursue efficient online algorithms for convex-shaped uncertainty sets. 
In the following sections, we focus on cases where $\tilde{\mathcal{P}}$ takes an exact polytopic or ellipsoidal form.
As demonstrated later in Section \ref{sct:simulations}, if the true convex hull is not of this form, it can still be closely over-approximated by a set of this form to yield a minimally conservative controller.

\begin{lemma}
    When $U = \R^m$, the robust optimization problem \eqref{eq:reformed robust optimization2} is feasible if $\tilde {\mathcal P} \cap \mathcal N = \emptyset$ where 
    $\mathcal N = \{ 0\}^m \times (-\infty,0]$ .
    \label{lem:feasible}
\end{lemma}
\begin{proof}
    Recall that $\mathcal N$ is convex and closed, and $\mathcal {\tilde P}$ is convex and compact. If they are disjoint, then the strict hyperplane separation theorem \cite{boyd2004convex} states that there exists a nonzero vector $(\mu,t) \in \mathbb R^{m+1}$ and $\gamma \in \mathbb R$ such that:
    \begin{equation*}
        \zeta_a^\top \mu + \zeta_b t > \gamma > \alpha t, \quad \forall (\zeta_a,\zeta_b) \in \mathcal{\tilde P}, \, \alpha \in (-\infty, 0].
    \end{equation*}
    For this to hold for all $\alpha \in (-\infty, 0]$, we must have $t \geq 0$ and $\gamma > 0$.
    Let $\varepsilon \in \mathbb R _{>0}$ be any constant such that:
    \begin{equation*}
        \varepsilon < \frac{\gamma}{| \inf_{(\zeta_a, \zeta_b) \in \mathcal{\tilde P}} \zeta_b|}.
    \end{equation*}
    We note that such an $\varepsilon$ will always exist since the infimum in the denominator is over a compact set. Now:
    $$\zeta_a ^\top\mu + \zeta_b(t+\varepsilon) > 0, \quad \forall (\zeta_a,\zeta_b) \in \mathcal{\tilde P},$$
    and hence the $u = \frac{1}{t+\varepsilon}\mu$ will be feasible for \eqref{eq:reformed robust optimization2}.
\end{proof}

Intuitively, Lemma \ref{lem:feasible} ensures that if the control gain $\zeta_a$ is sign-indefinite, the system's natural dynamics must inherently satisfy the CBF constraint ($\zeta_b \geq 0$ for all $(\zeta_a,\zeta_b) \in \mathcal{\tilde P}$).

For the sake of brevity, in the remainder of the discussion, we refer to $\tilde{\myset{P}}$ simply as the uncertainty set. We also define:
\begin{equation} \label{eq:inner_opt}
    m (u) := \min_{(\zeta_a, \zeta_b) \in \tilde{\myset{P}}} \zeta_a^\top u + \zeta_b.
\end{equation}
The safety controller in \eqref{eq:robust optimization} thus becomes:
\begin{equation} \label{eq:robust_optization_twolayers}
    k_{\rm ro}(\hat{x}, \myset{B}) = \argmin_{u\in U} f_0(u,\hat{x})  \ 
   \textup{ s.t. }  m(u) \geq 0.
\end{equation}
Thus, this robust optimization problem consists of the  inner optimization problem \eqref{eq:inner_opt} and the outer one \eqref{eq:robust_optization_twolayers}.
Furthermore, since $\tilde{\mathcal{P}}$ in \eqref{eq:inner_opt} depends on the current state estimate, this set $\tilde{\mathcal{P}}$ (or an over-approximation thereof) must be computed prior to every evaluation of \eqref{eq:robust_optization_twolayers}.

\subsection{Polytopic Uncertainty}

\begin{assumption}
    The set $\tilde{\myset{P}}$ forms a compact polytope, i.e:
    \begin{equation} \label{eq:polytopic_uncertainty}
        \tilde{\myset{P}} = \left\{ (\zeta_a,\zeta_b)\in \mathbb{R}^{m+1}: C\begin{bmatrix}
            \zeta_a \\
            \zeta_b
        \end{bmatrix} \leq d\right\},
    \end{equation}
    where $C\in \mathbb{R}^{p\times (m+1)}$ and $d \in \mathbb{R}^p$ are known.
    \label{ass:polytopic}
\end{assumption}

 Under Assumption \ref{ass:polytopic}, \eqref{eq:inner_opt} becomes:
\begin{equation} \label{eq:second_layer_opt}
\begin{aligned}
m(u) 
           & =  \min_{\eta \in \mathbb{R}^{m+1} } \begin{bmatrix}
         u^\top &   1
     \end{bmatrix} \eta \ \textup{ s.t. } \ C\eta \leq d,
\end{aligned}
\end{equation}
which is a linear program. 
Following the standard Lagrangian duality scheme \cite{boyd2004convex}, its dual problem is: 
\begin{equation} \label{eq:LP_dual}
\begin{aligned}
   m^{\star}(u) = \max_{\lambda\in \mathbb{R}^p} & \quad  - d^\top \lambda \\
    \textup{ s.t. } & C^\top \lambda + \begin{bmatrix}
        u \\
        1
    \end{bmatrix} = 0,  \lambda \geq 0.
\end{aligned}
\end{equation}
Since strong duality holds for linear programs \cite{boyd2004convex}, we know $m(u) =  m^{\star}(u)$. Thus, the original constraint $m(u)\geq 0$ in \eqref{eq:robust_optization_twolayers} is feasible if and only if there exists a $\lambda\in \mathbb{R}^p$ that satisfies the above constraints and  $- d^\top \lambda \geq 0$. We thus derive a safe controller as:
\begin{equation} \label{eq:equivalent QP}
    \begin{aligned}
            (k_{re}(\hat{x}, \myset{B}), \lambda(\hat{x}, \myset{B})) & =  \quad    \argmin_{u\in U, \lambda\in \mathbb{R}^p} \quad   \left\Vert u - k_d(\hat{x})  \right\Vert^{2}  \\
           \textup{ s.t. } &   - d^\top \lambda \geq 0,  C^\top \lambda +\begin{bmatrix}
        u \\
        1
    \end{bmatrix} = 0,  \lambda \geq 0.
    \end{aligned}
\end{equation}
This is again a quadratic program that can be efficiently solved online. Moreover, it is equivalent to \eqref{eq:robust_optization_twolayers} as summarized below.

\begin{proposition}
    Under Assumption \ref{ass:polytopic}, if $ \myset{U}_{ro} $ is not empty, then  $ k_{re}(\hat{x}, \myset{B}) $ in \eqref{eq:equivalent QP} exists and $k_{re}(\hat{x}, \myset{B})  = k_{\rm ro}(\hat{x}, \myset{B})$.
\end{proposition}

\subsection{Ellipsoidal Uncertainty}

\begin{assumption}
    The set $\tilde{\myset{P}}$ forms an ellipsoid, i.e.:
\begin{equation} \label{eq:ellipsoid_uncertainty}
    \tilde{\myset{P}} = \left\{ \eta = (\zeta_a,\zeta_b)\in \mathbb{R}^{m+1}: \eta^\top P \eta + q^\top \eta + r \leq 0 \right\},
\end{equation}
where $P\in \mathbb{R}^{(m+1)\times (m+1)}, q \in \mathbb{R}^{m+1}, r\in \mathbb{R}$ are known, and $ P = P^\top$ is positive definite.
\label{ass:ellipsoid}
\end{assumption}

Under Assumption \ref{ass:ellipsoid}, the inner layer optimization is:
\begin{equation} \label{eq:second_layer_opt_ellipsoid}
\begin{aligned}
m(u) & = 
\min_{\eta \in \mathbb{R}^{m+1} } \begin{bmatrix}
         u^\top &   1
     \end{bmatrix} \eta \\
     & \quad \textup{ s.t. } \  \eta^\top P \eta + q^\top \eta + r \leq 0.
\end{aligned}
\end{equation}

Following a similar procedure, the Lagrangian \cite{boyd2004convex} is:
\begin{equation}
    L(\eta,\lambda) = \begin{bmatrix}
         u^\top &   1
     \end{bmatrix} \eta + \lambda(\eta^\top P \eta + q^\top \eta + r),
\end{equation}
and the dual function is $l(\lambda) = \inf_{\eta\in \mathbb{R}^{m+1}} L(\eta,\lambda).$
If $\lambda=0$, the $L(\eta,\lambda)$ reduces to a non-zero linear function of $\eta$, yielding $l(0) = -\infty$. 
Since the dual problem attempts to maximize $l(\lambda)$, we can ignore this case and assume $\lambda > 0$.
Now, $L(\eta,\lambda)$ is quadratic in $\eta$ where $P$ is positive definite, hence the infimum in $l(\lambda)$ is achieved at $\nabla_{\eta} L(\eta,\lambda) = 0$.
This gives: 
\begin{equation}
    l(\lambda) = \lambda r - \frac{1}{4} ( \begin{bmatrix}
         u \\
         1
     \end{bmatrix}  + \lambda q )^\top (\lambda P)^{-1}( \begin{bmatrix}
         u \\
         1
     \end{bmatrix}  + \lambda q ), 
\end{equation}
and the dual problem:
\begin{equation*}
    \begin{aligned}
        m^{\star}(u) 
        = \max_{\lambda\in \mathbb{R}} & \   r \lambda - \frac{1}{4} ( \begin{bmatrix}
         u \\
         1
     \end{bmatrix}  + \lambda q )^\top (\lambda P)^{-1}( \begin{bmatrix}
         u \\
         1
     \end{bmatrix}  + \lambda q ) \\
    \textup{ s.t. } &  \lambda \geq 0 \\
     = \max_{\lambda, t \in \mathbb{R}} & \   r \lambda  - t \\
    \textup{ s.t. } &  \lambda \geq 0 \\
     & t - \frac{1}{4} ( \begin{bmatrix}
         u \\
         1
     \end{bmatrix}  + \lambda q )^\top (\lambda P)^{-1}( \begin{bmatrix}
         u \\
         1
     \end{bmatrix}  + \lambda q ) \geq 0.
    \end{aligned}
\end{equation*}

The last constraint can be equivalently expressed as a Linear Matrix Inequality (LMI) using the Schur complement:
\begin{equation} \label{eq:LMI condition}
  \begin{bmatrix}
         4t & [u^\top \ 1] + \lambda q^\top \\
         \begin{bmatrix}
             u \\
             1
         \end{bmatrix} + \lambda q & \lambda P
     \end{bmatrix} \succeq 0.
\end{equation}

Since the primal problem \eqref{eq:second_layer_opt_ellipsoid} has linear cost with a non-empty ellipsoidal feasible set, Slater's condition holds \cite{boyd2004convex}. Thus, strong duality holds, i.e., for any  $u$, $m^\star(u) =m(u)$. 

Now consider the outer layer of the optimization problem. Recall that the cost function is $f_0(u,\hat{x})  = \left\Vert u - k_d(\hat{x})  \right\Vert^{2}$. 
By introducing an auxiliary variable $s$, 
minimizing $f_0(u,\hat{x}) $ is equivalent to minimizing $s\in \mathbb{R}$ such that  $\left\Vert u - k_d(\hat{x})  \right\Vert^{2} \leq s$. Using the Schur Complement again, we  obtain an equivalent LMI constraint:
\begin{equation} \label{eq:LMI cost}
    \begin{bmatrix} 
         s & (u - k_d(\hat{x}))^\top  \\
        u - k_d(\hat{x}) & I
     \end{bmatrix} \succeq 0.
\end{equation}

 The reformulated controller is thus given by the following semi-definite program (SDP) in its standard form:
 \begin{equation} \label{eq:equivalent SDP}
    \begin{aligned}
            k_{re}(\hat{x}, \myset{B}), \lambda^\star, s^\star, t^\star & =  \quad    \argmin_{u\in U, \lambda,  s, t \in \mathbb{R}} \quad   s\\
           \textup{ s.t. } & 
           \eqref{eq:LMI condition}, \eqref{eq:LMI cost}, r\lambda - t \geq 0, \lambda \geq 0.
    \end{aligned}
\end{equation}

This equivalence between the above SDP and the robust optimization problem in \eqref{eq:robust_optization_twolayers} is summarized below.
\begin{proposition}
    Under Assumption \ref{ass:ellipsoid}, if $ \myset{U}_{ro} $ is not empty, then  $ k_{re}(\hat{x}, \myset{B}) $ in \eqref{eq:equivalent SDP} exists and $k_{re}(\hat{x}, \myset{B})  = k_{\rm ro}(\hat{x}, \myset{B})$.
\end{proposition}

\begin{remark}
    In general, the uncertainty set $\tilde{P}$, though convex, is hardly polytopic or ellipsoidal due to its nonlinear dependence on $\xi$. 
    In such scenarios, one needs to construct a polytopic or ellipsoidal over-approximation set $\hat{\myset{P}} \supseteq \tilde{\myset{P}}$ and then apply the aforementioned reformulations.
    Since $\hat{\myset{P}}$ contains all possible realizations of the CBF parameters, solving \eqref{eq:equivalent QP} and \eqref{eq:equivalent SDP} for $\hat{\myset{P}}$ provides a sufficient condition for robust safety. The conservativeness of this approach thus depends on how tight this over-approximation is. 
    The two examples shown in Section \ref{sct:simulations} describe how this can be achieved in practice.
\end{remark}

\section{Simulations}

\label{sct:simulations}

In this section, we provide two illustrative examples\footnote{https://github.com/rahalnanayakkara/duality-based-safety-filters} where we apply different methods for over-approximating the uncertainty set using polytopes and ellipsoids, respectively.

\begin{example}
\label{ex:double int}
    Consider a double integrator system with state $x = (x_1, x_2)\in \mathbb{R}^2$ and input $u\in \mathbb{R}$:
    \begin{equation*}
        \begin{bmatrix}
            \dot{x}_1 \\
            \dot{x}_2
        \end{bmatrix} = \begin{bmatrix}
            x_2 \\
            0
        \end{bmatrix} + \begin{bmatrix}
            0 \\
            1
        \end{bmatrix} u,
    \end{equation*}
    and the safety constraint
    $h(x) = 1 - x_1^2 - x_2^2-x_1x_2 \geq 0$. It can be verified that $h$ is a CBF for the linear class $\mathcal{K}$ function $\alpha(s) = s$. Following the convention in \eqref{eq:robust optimization}, we compute $a(x) = -x_1 -2x_2$ and $b(x) = 1 - x_1^2 - 2x_2^2 - 3x_1 x_2 $. 

    Suppose that an online state estimate $\hat{x}$ and an error bound set $\myset{B} = \{(e_1, e_2) \in \R^2: |e_1|\leq \delta, |e_2| \leq \delta \}$, where \mbox{$\delta \in \mathbb R_{>0}$}, are available for online feedback. 
    Denote the rectangular state uncertainty set $\{ \hat x\} \oplus \mathcal B$ as $\myset{I}_1 \times \myset{I}_2$.
    In the following, we give details about the uncertainty set $\myset{P}$, its convex over-approximation, as well as the safety filter design.
    
\begin{figure}[t]
        \centering
        \subfloat[Uncertainty Set]{
            \includegraphics[width=0.47\linewidth]{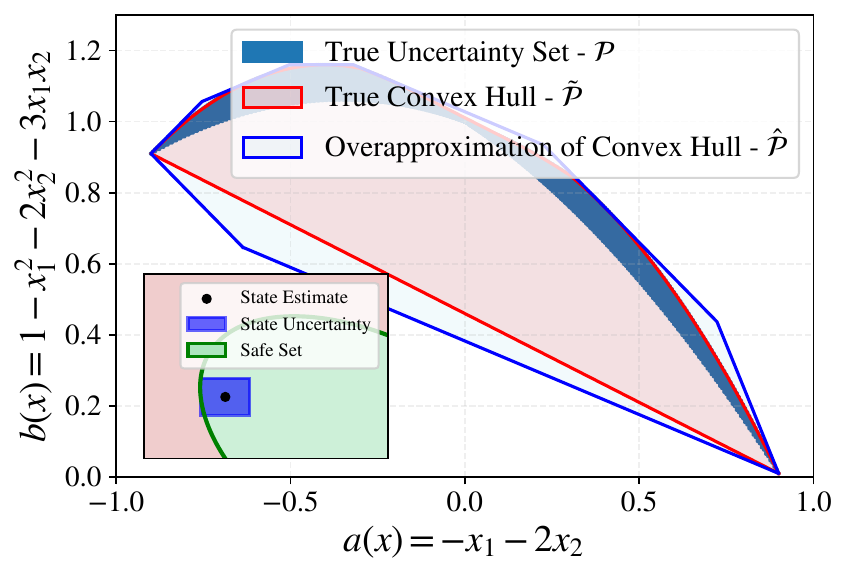}
            \label{fig:ex2_uncertain_set_a}
        }\hfill
        \subfloat[Degradation of safety with $\delta$]{
            \includegraphics[width=0.47\linewidth]{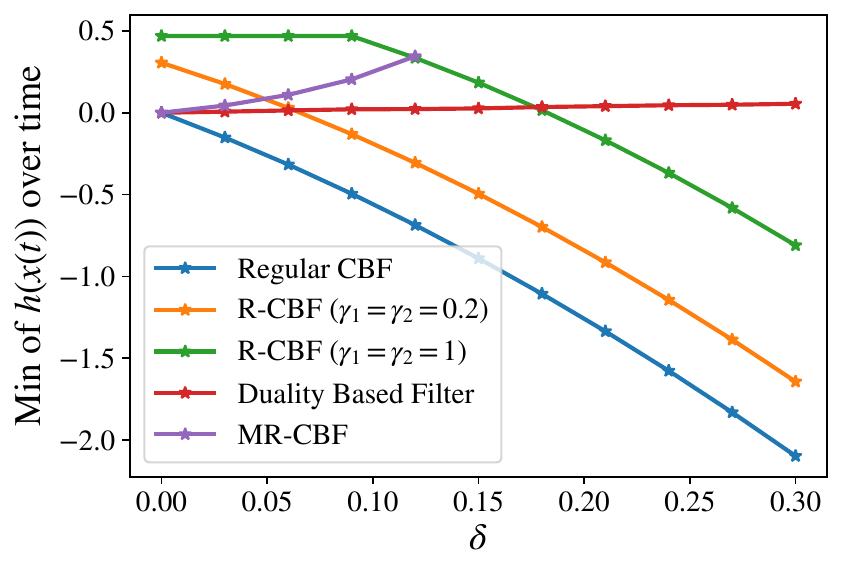}
            \label{fig:ex2_uncertain_set_b}
        }
        \caption{Visualization of the uncertainty set and the degradation of safety with $\delta$ for Example \ref{ex:double int}.}
        \label{fig:ex2_uncertain_set}
        \vspace{-10pt}
    \end{figure}

    \begin{figure}[t]
        \centering
        \subfloat[Comparison of trajectories]{
            \includegraphics[width=0.47\linewidth]{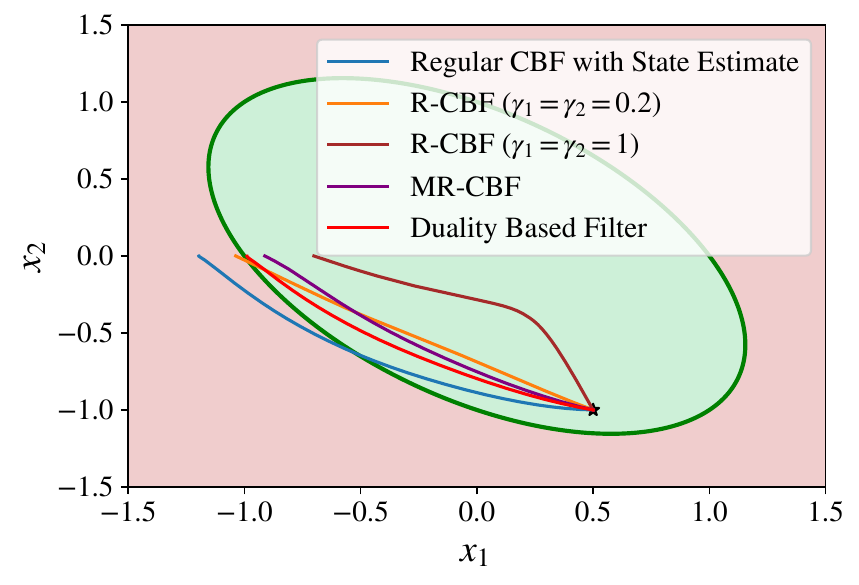}
            \label{fig:ex2_single_init_a}
        }\hfill
        \subfloat[Comparison of control input]{
            \includegraphics[width=0.47\linewidth]{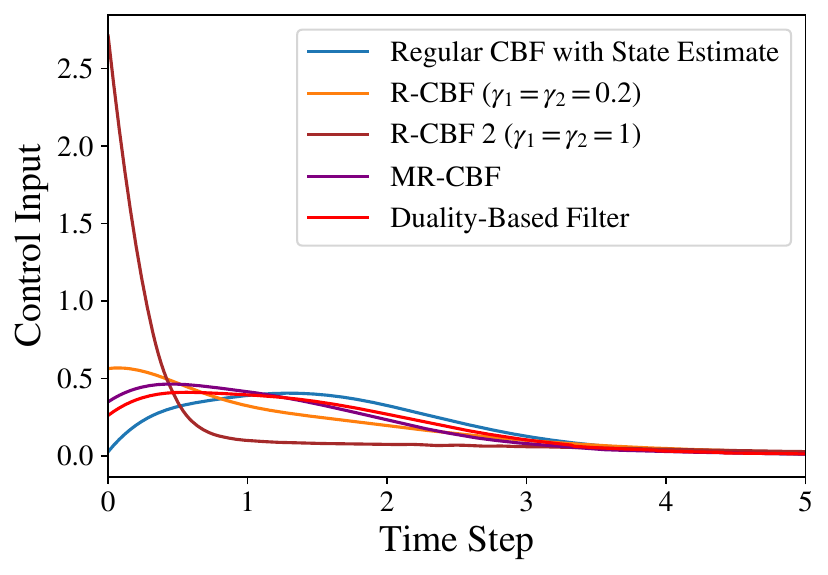}
            \label{fig:ex2_single_init_b}
        }
        \caption{Comparison of standard CBFs, R-CBFs, MR-CBFs and the proposed duality based method for Example \ref{ex:double int}.}
        \label{fig:ex2_single_init}
        \vspace{-10pt}
    \end{figure}
    
    \textit{Over-approximation of the uncertainty set}: We use the supporting hyperplane technique \cite{frehse2013flowpipe} to compute an over-approximation of the convex hull of $\myset{P}$. For any unit direction vector \mbox{$v=(v_1, v_2) \in \mathbb{S}^1$}, we have:
    \begin{equation*}
        v^T \zeta \leq \sup_{\tilde{\zeta} \in \tilde{(\myset{P})}} v^T \tilde{\zeta}, \quad \forall \zeta \in \tilde{(\myset{P})}.
    \end{equation*}
    The supremum on the right hand side can be re-written as \mbox{$\sup_{(x_1, x_2) \in \myset{I}_1 \times \myset{I}_2} \phi_v(x)$}, where $\phi_v(x) = v_1 a(x) + v_2 b(x)$.
    Since $\phi_v(x)$ is an indefinite quadratic function of $(x_1, x_2)$ for all $v \in \mathbb{R}^2$, and $\myset{I}_1 \times \myset{I}_2$ is rectangular, this supremum will be attained on the boundary, and will have the value:
    \begin{equation*}
        \max \{\gamma_1(\hat x_1-\delta), \gamma_1(\hat x_1+\delta), \gamma_2(\hat x_2-\delta), \gamma_2(\hat x_2+\delta)\},
    \end{equation*}
    where $\gamma_1(z) = \sup_{x_2 \in \myset{I}_2} \phi((z, x_2))$, and $\gamma_2(z) = \sup_{x_1 \in \myset{I}_1} \phi((x_1, z)).$
    Since both $\gamma_1, \gamma_2$ are the supremum of a quadratic function of a single variable over a closed interval, they can be computed analytically.
    Thus, for any given vector $v \in \mathbb{S}^1$, a supporting hyperplane of the form $v^T \zeta \leq c$, where $c \in \mathbb{R}$, can be obtained. By repeating this procedure for multiple such vectors $\{v^{(1)}, \dots,v^{(p)}\}$, where $p \in \mathbb{Z}_{>0}$, we can construct a polytopic over-approximation of the convex hull of $\myset{P}$. Theoretically, as $p \to \infty$ we recover the exact convex hull $\tilde{\myset{P}}$, and thus in practice we may choose a sufficiently large $p$ such that a desired accuracy is achieved.
    
    We chose $p=16$ directions $v^{(i)}$ uniformly sampled from the unit circle. 
    Fig. \ref{fig:ex2_uncertain_set}(a) shows state-space uncertainty at a specific instant (inset) alongside the corresponding set $\myset{P}$, its convex hull $\tilde{\myset{P}}$, and the hyperplane-based over-approximation $\hat{\myset{P}}$.
    We again note that any additional conservatism that may arise comes from the over-approximation of $\tilde{\myset{P}}$ by $\hat{\myset{P}}$ and not when going from $\myset{P}$ to $\tilde {\myset{P}}$.

    Simulations use OSQP over a $20$-second horizon with $0.02\,\text{s}$ steps. Fig. \ref{fig:ex2_single_init}(a) shows closed-loop trajectories from an initial state (star), where the reference $k_d = 0$ drives the system toward the unsafe region. Fig. \ref{fig:ex2_single_init}(b) shows the required inputs: R-CBFs with large $\gamma_1, \gamma_2$ demand large inputs and are overly conservative, unnecessarily restricting the system to the interior. Conversely, regular CBFs or R-CBFs with small $\gamma_1, \gamma_2$ violate safety. In contrast, our proposed method guarantees safety while minimizing conservatism.
    
    Fig. \ref{fig:ex2_uncertain_set}(b) depicts the minimum $h(x(t))$ across 10 initial conditions around the origin under varying uncertainty magnitudes $\delta$. Regular CBFs fail ($h(x)<0$) for any $\delta>0$. R-CBFs ensure safety for small $\delta$ but remain highly conservative. MR-CBFs are similarly conservative and become infeasible (i.e., \eqref{eq:MR-CBF} lacks a solution) for $\delta > 0.1$. In contrast, our method consistently maintains $h(x)\geq 0$ for all tested $\delta$ with minimal conservatism.
\end{example}

\begin{example}
\label{ex:segway}
We consider a Segway model: 
\begin{equation}
    \begin{aligned}
        \begin{bmatrix}
            \dot{p}(t) \\
            \dot{\phi}(t) \\
            \dot{v}(t) \\
            \dot{\omega}(t)
        \end{bmatrix} = 
        \begin{bmatrix}
            v \\
            \omega \\
            f_{v}(\phi,v,\omega) \\
            f_{w}(\phi,v,\omega)
        \end{bmatrix} +  
        \begin{bmatrix}
            0 \\
            0 \\
            g_{v}(\phi) \\
            g_{w}(\phi)
        \end{bmatrix} u ,
    \end{aligned}
\end{equation}
with its state $x =(p, \phi,v, \omega)$, where $p, v, \phi, \omega$ denote the position and velocity of the wheel, the pitch angle and angular velocity of its upper body, respectively. The control input $u\in [-8, 8]$ is the voltage applied on the motors at the wheels. Due to space limitations, we refer the interested reader to \cite{molnar2022safety} for more details about this model.

For simplicity of illustration, we consider a safety constraint on rotational movement: $h(x) = 1 - (3\phi^2 + 2\phi
\omega + \omega^2)$ similar to \cite{nanayakkara2025safety}, and simulate the model starting from an initial condition $x_0 = (-4, -0.5, 0, 1)$. The online state estimate $\hat{x}$ and state uncertainty $\myset{B} = \{z: \| z\|\leq \varepsilon \}$ are assumed to be known for feedback design. 
The nominal input is an unsafe LQR controller $u = K^\top \hat{x}$. We inject an adversarial error 
$\hat{x} = x - \frac{\epsilon}{\sqrt{x_2^2+x_4^2}} [0, x_2, 0, x_4]^\top$,
to deliberately mislead the controller by making $h(\hat{x}) \geq h(x)$ (cf. Fig \ref{fig:segway example}(b)).

In this example, we numerically over-approximate the uncertainty set in the safety filter using an ellipsoid.
We first take $N=1000$ state samples $\{\tilde{x}^i\} $ from the state uncertainty set, compute their respective CBF parameters $(a(\tilde{x}^i), b(\tilde{x}^i))$,  and then numerically find an enclosing ellipsoid using the mean and covariance matrix of the samples. 
The reformulated controller \eqref{eq:equivalent SDP} is solved via CVXPY's CLARABEL over a $10$-second horizon with $0.02\,\text{s}$ steps.

In Fig \ref{fig:segway example} (a) we compare our approach against naively using a regular CBF with the state estimate and the R-CBF from \cite{nanayakkara2025safety}. 
Since the R-CBF values for the constants $\gamma_1$ and $\gamma_2$, taken from \cite{nanayakkara2025safety} were tuned for the specific level of uncertainty present, both approaches guarantee invariance of the safe set.
We also compare the performance of the proposed approach for different levels of state uncertainty in Fig \ref{fig:segway example} (c), and observe that the trajectories always stay within the safe region. 

\begin{figure}[t] 
    \centering
    \subfloat[Comparison of safety filters]{
        \includegraphics[width=0.305\linewidth]{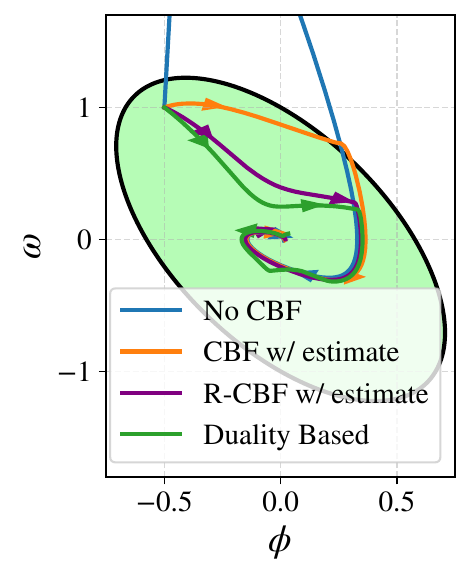}
    }\hfill
    \subfloat[True state vs.\ state estimate]{
        \includegraphics[width=0.305\linewidth]{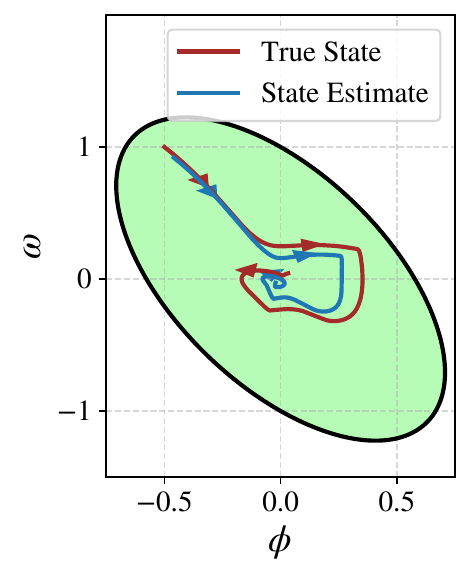}
    }\hfill
    \subfloat[Varying levels of state uncertainty]{
        \includegraphics[width=0.305\linewidth]{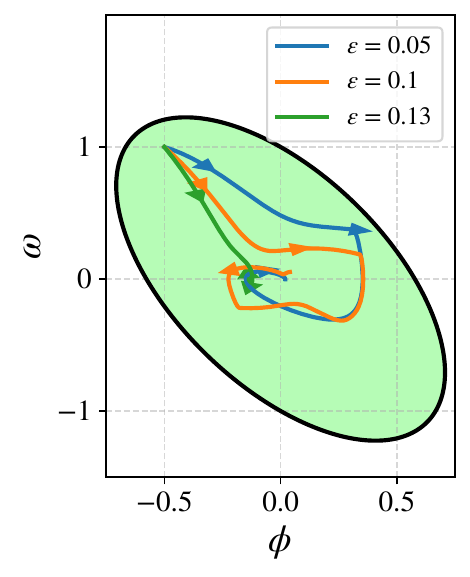}
    }
    \caption{State trajectories for the segway in Example \ref{ex:segway}.}
    \label{fig:segway example}
\vspace{-10pt}
\end{figure}

\end{example}

\section{Conclusion} 
\label{sec:conclusion}

In this work, we proposed a duality-based reformulation for robust optimization problems arising from CBF-based safety filter design under state uncertainty. We first demonstrated that replacing the image of the state uncertainty under the CBF coefficient maps by its convex hull does not introduce conservativeness into the robust optimization problem. 
Furthermore, we established that when this convex hull is polytopic or ellipsoidal, our duality-based approach yields tractable, equivalent reformulations of the original problem. Through two case studies, we demonstrated that even when the convex hull is not of this form, bounding it with polytopic or ellipsoidal over-approximations effectively eliminates the need for Lipschitz constant estimates and provides a less conservative solution than prior methods.
Although the requirements for feasibility were discussed, guaranteeing the recursive feasibility of the closed-loop system remains an open problem and will be the subject of future research.

\bibliographystyle{IEEEtran}
\bibliography{IEEEabrv,references}

\end{document}